\documentclass[12pt,3p]{nj_style}

\usepackage{amsmath}
\usepackage{amssymb}
\usepackage{amsthm}
\usepackage{enumerate}
\usepackage[font=small,labelfont=bf,margin=8pt]{caption}
\usepackage[usenames,dvipsnames]{color}
\usepackage{graphicx}
\usepackage[colorlinks=true]{hyperref}

\usepackage{bbm} 

\newtheorem{thm}{Theorem}

\newtheorem{lemma}[thm]{Lemma}
\newtheorem{prop}[thm]{Proposition}

\newenvironment{numberedprop}[1]
{
  \renewcommand{\thethm}{#1}
  \begin{prop}
}{
  \end{prop}
  \addtocounter{thm}{-1}
}

\newcommand{\bra}[1]{\langle #1 |}
\newcommand{\ket}[1]{| #1 \rangle}

\newcommand{\cl}[1]{\mathcal{#1}}

\newcommand{\mnorm}[1]{%
\left\vert\kern-0.9pt\left\vert\kern-0.9pt\left\vert #1
\right\vert\kern-0.9pt\right\vert\kern-0.9pt\right\vert}
\newcommand{\bmnorm}[1]{%
\big\vert\kern-0.9pt\big\vert\kern-0.9pt\big\vert #1
\big\vert\kern-0.9pt\big\vert\kern-0.9pt\big\vert}

\newcommand{\pa}[1]{(#1)}
\newcommand{\Pa}[1]{\left(#1\right)}

\newcommand{\set}[1]{\{#1\}}
\newcommand{\Set}[1]{\left\{#1\right\}}

\newcommand{\inner}[2]{\langle #1 , #2\rangle}
\newcommand{\Inner}[2]{\left\langle #1 , #2\right\rangle}

\DeclareMathOperator{\trace}{Tr}
\newcommand{\ptr}[2]{\trace_{#1}\pa{#2}}
\newcommand{\Ptr}[2]{\trace_{#1}\Pa{#2}}

\newcommand{\tinyspace}{\mspace{1mu}}

\newcommand{\norm}[1]{\lVert\tinyspace#1\tinyspace\rVert}

\newcommand{\dnorm}[1]{\norm{#1}_{\diamond}}

\newcommand{\fontmapset}{\mathrm} 

\newcommand{\mset}[2]{\fontmapset{#1}\pa{#2}}

\newcommand{\lin}[1]{\mset{L}{#1}}

\newcommand{\zpt}{\bQ}
\newcommand{\vaa}{\bV}

\newcommand{\identity}{\mathbbm{1}}
\newcommand{\idsup}[1]{\identity_{#1}}

\newcommand{\trans}{{\scriptstyle\mathsf{T}}} 
\newcommand{\defeq}{\stackrel{\smash{\textnormal{\tiny def}}}{=}}

\DeclareMathOperator{\rank}{rank}
\DeclareMathOperator{\spn}{span}

\def\ot{\otimes}

\def\bbC{\mathbb{C}}

\def\bbM{\mathbb{M}}

\def\bbR{\mathbb{R}}

\def\cX{\mathcal{X}}
\def\cY{\mathcal{Y}}
\def\cZ{\mathcal{Z}}

\def\bQ{\mathbf{Q}}

\def\bV{\mathbf{V}}

\def\thmlowrank
{
  If $d^2 \geq 2q+1 \geq 3$ then there exists a subspace $\bV_{2q}\subset\lin{\cY\cX}$ of $d^2\times d^2$ Hermitian matrices of dimension
  \[
    \dim(\bV_{2q}) =
    \begin{cases}
     (d^2 - 2q)^2 - (d-\lfloor (2q+d-1)/d \rfloor)^2 & \textrm{if $d$ is odd} \\
     (d^2 - 2q)^2 - (d-\lfloor (2q+d-2)/d \rfloor)^2 & \textrm{otherwise}
    \end{cases}
  \]
  such that every nonzero $V\in\bV_{2q}$ has rank at least $2q+1$ and $\ptr{\cY}{V}=0_\cX$.
}

\def\thmunitary
{
  If $d \geq 3$ then there exists a subspace $\bV_{2,\mathrm{unital}} \subset \lin{\cY\cX}$ of $d^2\times d^2$ Hermitian matrices of dimension
  \[ \dim(\bV_{2,\mathrm{unital}}) = d^4 - 6d^2 + 2d + 5 \]
  such that every nonzero $V\in\bV_{2,\mathrm{unital}}$ has rank at least $3$ and $\ptr{\cY}{V}=\ptr{\cX}{V}=0$.
  If $d = 2$ then there exists such a subspace of dimension $3$.
}

\def\thmall
{
  If $d^2 \geq 2q+2 \geq 2$ then there exists a subspace $\bV_{q+}\subset\lin{\cY\cX}$ of $d^2\times d^2$ Hermitian matrices of dimension
  \[
    \dim(\bV_{q+}) =
    \begin{cases}
      d^4 - (4q+1)d^2 + (4q^2 + 2q) & \textrm{if } q \geq d \\
	  d^4 - (4q+1)d^2 + (4q^2 + 2q) - (d-q)(d-q-1) & \textrm{otherwise}
    \end{cases}
  \]
  such that every nonzero $V\in\bV_{q+}$ has at least $q+1$ positive eigenvalues and $\ptr{\cY}{V}=0_\cX$.
}

\def\thmunital
{
  If $d^2 \geq 2q+2 \geq 2$ then there exists a subspace $\bV_{q+,\mathrm{unital}}\subset\lin{\cY\cX}$ of $d^2\times d^2$ Hermitian matrices of dimension
  \[ \dim(\bV_{q+,\mathrm{unital}}) \geq \dim(\bV_{q+}) - d^2 + d + 2, \]
  where $\dim(\bV_{q+})$ is as in Proposition~\ref{thm:all}, such that every nonzero $V\in\bV_{q+,\mathrm{unital}}$ has at least $q+1$ positive eigenvalues and $\ptr{\cY}{V}=\ptr{\cX}{V}=0$.
}

\begin{document}

\begin{frontmatter}

\title{
Process tomography for unitary quantum channels
}

\author[PI]{Gus Gutoski}
\ead{ggutoski@perimeterinstitute.ca}
\author[IQC]{Nathaniel Johnston}
\ead{nathaniel.johnston@uwaterloo.ca}

\address[PI]{Perimeter Institute for Theoretical Physics, Waterloo, Ontario, Canada}
\address[IQC]{Institute for Quantum Computing, University of Waterloo, Waterloo, Ontario, Canada}

\renewcommand*{\today}{October 11, 2013\\{\small (Minor revisions: March 25, 2014)}}

\begin{abstract}
We study the number of measurements required for quantum process tomography under prior information, such as a promise that the unknown channel is unitary.
We introduce the notion of an interactive observable and we show that any unitary channel acting on a $d$-level quantum system can be uniquely identified among all other channels (unitary or otherwise) with only $O(d^2)$ interactive observables, as opposed to the $O(d^4)$ required for tomography of arbitrary channels.
This result generalizes, so that channels with at most $q$ Kraus operators can be identified with only $O(qd^2)$ interactive observables.
Slight improvements can be obtained if we wish to identify such a channel only among unital channels or among other channels with $q$ Kraus operators.
These results are proven via explicit construction of large subspaces of Hermitian matrices with various conditions on rank, eigenvalues, and partial trace.
Our constructions are built upon various forms of totally nonsingular matrices.
\end{abstract}


\end{frontmatter}

\section{Introduction}
\label{sec:intro}

The problem of deducing the action of an unknown quantum channel by gathering statistics from repeated measurement is called \emph{process tomography}.
In general, one must measure a set of $O(d^4)$ distinct observables in order to uniquely identify a given channel $\Phi$ acting on a $d$-level quantum system, owing to the fact that the set of all such channels spans an affine space of dimension $d^4 - d^2$.

However, it could be that far fewer observables are required if we are promised that $\Phi$ belongs to some prescribed set $A$ of channels or if we only wish to identify $\Phi$ among some specific subset $B$ of all possible channels.
In other words, we are interested in the following question for various sets $A \subseteq B$ of channels:
\begin{quote}
  \emph{Given an unknown channel $\Phi\in A$, what is the minimum number of observables required to ensure that there is no other channel in $B$ with the same measurement statistics as $\Phi$?}
\end{quote}

\subsection{Primary results}
\label{sec:intro:results}

We show that if $\Phi$ is promised to be a \emph{unitary} channel then it suffices to measure only $O(d^2)$ observables in order to uniquely identify $\Phi$ among all channels, and that slight improvements can be obtained if we wish to identify $\Phi$ only among unital or unitary channels.
In particular,
\begin{enumerate}
\item \label{it:result:unitary} $4d^2-2d-4$ observables suffice to identify a unitary among other unitaries.
If $d=2$ then only 6 observables suffice.
\item $5d^2-3d-4$ observables suffice to identify a unitary among all channels.
\item $5d^2-4d-5$ observables suffice to identify a unitary among unital channels.
\end{enumerate}
These results generalize to the problem of identifying channels with at most $q$ Kraus operators; we show that it suffices to measure $O(qd^2)$ observables, with slight improvements if we wish to identify such a channel only among unital channels or among other channels with $q$ Kraus operators.
A detailed statement of these results is given in Theorem \ref{thm:unitary_uda_udp} of Section \ref{sec:results:tomography}.
Several potential extensions of these results are described in Section \ref{sec:future}.

\subsection{Prior work}

Our work is inspired by previous results on the number of observables required for \emph{state} tomography, the study of which was initiated by Heinosaari, Mazzarella, and Wolf \cite{HTW11} who showed that any pure state of a $d$-level quantum system can be uniquely identified among other pure states with only $O(d)$ observables, as opposed to the $O(d^2)$ required for tomography of an arbitrary state.
In particular, $4d-5$ observables suffice and this bound generalizes to the problem of identifying states with rank at most $q$ among other states with rank at most $q$.

Chen, Dawkins, Ji, Johnston, Kribs, Shultz, and Zeng investigated the number of observables required in order to identify a given pure state among \emph{all} states \cite{CDJJKSZ12}.
Those authors showed that $O(d)$ observables suffice even in this more demanding setting.
In particular, $4d-5$ observables no longer suffice but $5d-7$ observables is sufficient.
As in Ref.\ \cite{HTW11}, this bound generalizes to the problem of identifying states with rank at most $q$ among all states. This work was continued by Carmeli, Heinosaari, Schultz, and Toigo \cite{CHST13}, who slightly improved the number of observables that suffice in this setting in some cases.

At the time of this writing we were aware of prior results on the number of observables required for process tomography only under certain sparsity assumptions \cite{SKMRBAFW11}, but not under the promise that the unknown channel is unitary or has few Kraus operators.
Additional prior work was subsequently brought to our attention after the journal version of this article was published.

Most notably, Flammia, Gross, Liu, and Eisert used compressed sensing techniques to show that a channel with $q$ Kraus operators can be identified with $O(qd^2\log d)$ observables, which is already optimal to within a logarithmic factor of $d$ \cite{FGLE12}.
Their article contains a wealth of references on state tomography under prior information, and the idea for their approach to process tomography was described as early as 2009 \cite{GLFBE10}.
Moreover, their interactive observables can be implemented by preparing only eigenstates of Pauli operators and by measuring only Pauli observables, which tend to admit relatively easy implementations in the laboratory.

Along similar lines, Kimmel, da Silva, Ryan, Johnson, and Ohki showed that the unital part of any multi-qubit channel can be deduced from Clifford measurements and that their procedure is robust in the presence of measurement error \cite{KSRJO13}.

By contrast to this prior work, our primary goal is to minimize the number of observables required for tomography without regard to how those observables might be implemented.
However, in Section \ref{sec:future:experimental} we observe that the six observables required for tomography of qubit unitary channels are in fact products of Clifford operators.

\subsection{Methods and supplementary results}

The primary results of Section \ref{sec:intro:results} rest upon several supplementary results on interactive observables, on the construction of subspaces of Hermitian matrices with various conditions on rank, eigenvalues, and partial trace, and on the construction of totally nonsingular matrices with various useful properties.

\subsubsection{Interactive observables}

Before we can talk of the number of observables required for process tomography it is necessary to clarify the meaning of the word ``observable'' in this context.
In Section \ref{sec:observable} we introduce the notion of an \emph{interactive observable} for channels, a special case of which is an ordinary observable for states.
The latter specifies a measurement for a state; the former, both an input state for a channel and a measurement for the resulting output.

We claim that every $d^2\times d^2$ Hermitian matrix can be viewed as an interactive observable for channels acting on a $d$-level quantum system in the same sense that every $d\times d$ Hermitian matrix can be viewed as an ordinary observable for states of a $d$-level quantum system.

To this end we observe that an arbitrary $d^2\times d^2$ Hermitian matrix $H$ can be decomposed into an input-state--measurement pair if and only if $H$ lies in the unit ball of the dual of the diamond norm (also known as the completely bounded 1-norm or the completely bounded trace norm), from which the above claim follows via a simple scaling argument.
(See Section \ref{sec:observable:operational} for details and Section \ref{sec:future:strategies} for further generalization.)
By comparison, it is easily seen that an arbitrary $d\times d$ Hermitian matrix can be decomposed into a binary measurement if and only if it lies in the unit ball of the operator norm.

\subsubsection{Subspaces of matrices with rank, eigenvalue, and partial trace conditions}

In Section \ref{sec:reduction} we observe that questions about the number of interactive observables required for process tomography can be reduced to the problem of finding large subspaces of Hermitian matrices with some combination of restrictions on rank, eigenvalues, and partial trace.
Detailed statements of our primary result on process tomography and supplementary results on Hermitian subspaces can be found in Section \ref{sec:results}.
Explicit constructions of these subspaces appear in Section \ref{sec:high_rank}.

\subsubsection{Totally nonsingular matrices}

Our supplementary results on Hermitian subspaces in Section \ref{sec:high_rank} rest upon the existence of various forms of totally nonsingular matrices, proofs of which appear in Section \ref{sec:total_nonsing}.

\subsection{Notation, linear algebra, quantum formalism}

We restrict attention throughout the paper to finite-dimensional linear algebra.
The symbol $\bbM_d$ denotes the $d^2$-dimensional complex vector space of $d\times d$ matrices with complex entries, which is identified in the usual way with the space of linear operators acting on the $d$-dimensional complex Euclidean space $\bbC^d$.

It is helpful to give separate labels to the different factor spaces in tensor product spaces such as $\bbC^d\ot\bbC^d$ or $\bbM_d\ot\bbM_d$.
To this end the Calligraphic letters $\cX,\cY,\cZ$ are used to denote isomorphic copies of $\bbC^d$.
Tensor products such as $\cX\ot\cY$ are abbreviated to $\cX\cY$ so that
\[ \cX\cY = \cX\ot\cY \simeq \bbC^d\ot\bbC^d \simeq \bbC^{d^2}. \]
The symbol $\lin{\cX}$ denotes the vector space of linear operators acting on $\cX$ so that $\lin{\cX}\simeq\bbM_d$ and
\[ \lin{\cX\cY} = \lin{\cX\ot\cY} \simeq \bbM_d\ot\bbM_d \simeq \bbM_{d^2}. \]

Linear maps from matrices to matrices of the form $\Phi:\lin{\cX}\to\lin{\cY}$ are in one-to-one correspondence with elements of $\lin{\cY\cX}$ via the \emph{Choi--Jamio\l kowski isomorphism}:
\[ J(\Phi) = \sum_{i,j=1}^d \Phi(\ket{i}\bra{j}) \ot \ket{i}\bra{j}. \]
Here $\set{\ket{1},\dots,\ket{d}}$ denotes the standard orthonormal basis of $\bbC^d$ written in the ket notation.

A linear map $\Phi:\lin{\cX}\to\lin{\cY}$ is:
(a) completely positive if and only if $J(\Phi)$ is positive semidefinite,
(b) trace-preserving if and only if $\ptr{\cY}{J(\Phi)}=I_\cX$, and
(c) unital if and only if $\ptr{\cX}{J(\Phi)}=I_\cY$.
Moreover, $\Phi$ is completely positive if and only if there exist linear operators $A_1,\dots,A_q:\cX\to\cY$ with
\[ \Phi: X \mapsto \sum_{i=1}^q A_iXA_i^*. \]
Operators with this property are called \emph{Kraus operators} for $\Phi$.
The minimum number of Kraus operators required in any such decomposition of $\Phi$ is equal to the rank of the Choi matrix $J(\Phi)$.

Associated with each $d$-level quantum system is a $d$-dimensional complex Euclidean space $\cX$.
The set of possible \emph{states} of such a system is identified with the set of positive semidefinite matrices $\rho\in\lin{\cX}$ with trace equal to one.
The set of possible \emph{channels} acting on such a system is identified with the set of completely positive and trace-preserving linear maps $\Phi:\lin{\cX}\to\lin{\cY}$.
Each \emph{measurement} of such a system with a finite number of outcomes indexed by $a$ is represented by a finite set $\set{P_a}\subset\lin{\cX}$ of positive semidefinite matrices with $\sum_a P_a = I$.
The probability with which a system in state $\rho$ yields outcome $a$ when measured according to $\set{P_a}$ is given by the inner product $\ptr{}{P_a\rho}$.

\section{Interactive observables}\label{sec:observable}

The concept of an observable is familiar in the context of state tomography: every $d\times d$ Hermitian matrix $H$ specifies an \emph{observable} for a $d$-level quantum system.
If that system is in state $\rho$ then the \emph{expectation of measuring $H$} is the quantity $\ptr{}{H\rho}$.

Perhaps less widely known is the fact that every $d^2\times d^2$ Hermitian matrix $H\in\lin{\cY\cX}$ can also specify an ``observable'' for a \emph{channel} acting on a $d$-level quantum system.
Given such a channel $\Phi:\lin{\cX}\to\lin{\cY}$ the \emph{expectation of measuring $H$} is the quantity $\ptr{}{HJ(\Phi)}$---the inner product between $H$ and the Choi matrix $J(\Phi)$ associated with $\Phi$.
In the context of process tomography, observables such as $H$ shall be called \emph{interactive observables} in order to distinguish them from observables for states.

\subsection{Operational interpretation of interactive observables}
\label{sec:observable:operational}

What does it mean to measure an interactive observable $H$ for a channel $\Phi$?
To answer this question it is helpful to recall the meaning of measurement of an ordinary observable $H\in\lin{\cX}$ on a state $\rho\in\lin{\cX}$.
If $\norm{H}\leq 1$ then it is easy to construct a binary measurement $\set{P_\pm}\subset\lin{\cX}$ with $H=P_+-P_-$.
Suppose a system in state $\rho$ is measured according to $\set{P_\pm}$ and suppose that we assign the quantities  $\pm 1$ to the measurement outcomes $\pm$.
Then the expected value of this quantity is $\ptr{}{P_+\rho}-\ptr{}{P_-\rho}=\ptr{}{H\rho}$---the expectation of measuring $H$.
If $\norm{H}> 1$ then the expectation of measuring $H$ is simply $\norm{H}$ times the expectation of measuring $H/\norm{H}$, which, again, is $\ptr{}{H\rho}$.

For interactive observables on channels, if $H\in\lin{\cY\cX}$ is sufficiently small then a more elaborate construction can be used to extract from $H$ a state $\xi\in\lin{\cX\cZ}$ and a binary measurement $\set{P_\pm}\subset\lin{\cY\cZ}$ with the property that
\[ \ptr{}{HJ(\Phi)} = \Ptr{}{P_+ (\Phi\ot\idsup{\cZ})(\xi)} - \Ptr{}{P_- (\Phi\ot\idsup{\cZ})(\xi)}. \]
Thus, if we apply $\Phi$ to one of two $d$-level systems in joint state $\xi$ and measure the resulting state $(\Phi\ot\idsup{\cZ})(\xi)$ according to $\set{P_\pm}$, assigning the quantities $\pm 1$ to the measurement outcomes $\pm$ as before, then the expected value of this quantity is precisely the inner product $\ptr{}{HJ(\Phi)}$ of $H$ with the Choi matrix $J(\Phi)$---the expectation of measuring $H$.

\subsection{How to decompose an interactive observable} \label{sec:observable:decompose}

Given an interactive observable $H$, how does one compute the associated state $\xi$ and measurement $\set{P_\pm}$?
Suppose $H\in\lin{\cY\cX}$ is small enough that it can be written $H=Q_+-Q_-$ for some positive matrices $Q_\pm$ such that $Q_++Q_-=I_\cY\ot\rho$ for some state $\rho\in\lin{\cX}$.
Such a set $\set{Q_\pm}$ has been called a \emph{one-round measuring co-strategy} \cite{GutoskiW07}, a \emph{1-tester} \cite{ChiribellaD+09a}, a \emph{process POVM} \cite{Ziman08}, and an \emph{interactive measurement} \cite{JainU+09}.
Each of Refs.\ \cite{GutoskiW07,ChiribellaD+09a,Ziman08} offers a proof that any interactive measurement $\set{Q_\pm}\subset\lin{\cY\cX}$ can be decomposed into a state $\xi\in\lin{\cX\cZ}$ and a measurement $\set{P_\pm}\subset\lin{\cY\cZ}$ with
\[ \Ptr{}{Q_\pm J(\Phi)} = \Ptr{}{P_\pm (\Phi\ot\idsup{\cZ})(\xi)} \]
for all channels $\Phi:\lin{\cX}\to\lin{\cY}$.
These proofs are constructive and the construction can be implemented efficiently on a computer.

We are not aware of a succinct formula for $\xi,\set{P_\pm}$ in terms of an arbitrary interactive observable $H$.
But such a formula can be derived for the special case in which the state $\rho\in\lin{\cX}$ above is the completely mixed state $I_\cX/d$.
In this case we may take
\[ \xi=\ket{\phi^+}\bra{\phi^+} \]
where $\ket{\phi^+}=\frac{1}{\sqrt{d}}\sum_{i=1}^d \ket{i}_\cX\ket{i}_\cZ$ denotes the canonical maximally entangled pure state.
If we let $H_\pm$ denote the positive and negative parts of $H$ then the interactive measurement $\set{Q_\pm}$ for $H$ can be written
\[ Q_\pm = H_\pm + \frac{I_{\cY\cX}/d-H_+-H_-}{2} \]
and the desired measurement $\set{P_\pm}$ is given by normalization and matrix transposition: \[ P_\pm=dQ_\pm^\trans. \]

\subsection{Our interactive observables may be implemented with a maximally entangled state}

Our results on process tomography are not sensitive to scaling.
Specifically, the sets of interactive observables presented in the present paper for the purpose of identifying unitary channels can always be assumed to consist entirely of interactive observables $H$ with $\norm{H}\leq 1/d$ so that the formulas of Section \ref{sec:observable:decompose} applies for extracting $\xi,\set{P_\pm}$ from $H$.
In particular, each of our interactive observables can be implemented by applying $\Phi$ to the maximally entangled pure state $\ket{\phi^+}$---the only difference among our observables is in how the resulting state $(\Phi\ot\idsup{\cZ})(\ket{\phi^+}\bra{\phi^+})$ is measured.

In retrospect it is not surprising that a fixed maximally entangled pure input state suffices for process tomography:
every channel $\Phi$ is uniquely determined by its Choi matrix $J(\Phi)$, which is given by
\[J(\Phi)= d(\Phi\ot\idsup{\cZ})(\ket{\phi^+}\bra{\phi^+}).\]

\subsection{The dual of the diamond norm is the appropriate metric for interactive observables}

In Section \ref{sec:observable:decompose} we claimed that a small enough Hermitian matrix $H\in\lin{\cY\cX}$ can always be
decomposed into a state $\xi\in\lin{\cX\cZ}$ and measurement $\set{P_\pm}\subset\lin{\cY\cZ}$.
The curious reader might ask, ``How small is `small enough'?''
The answer is that $H$ must lie in the unit ball of the \emph{dual of the diamond norm}.

Specifically, it was shown via semidefinite programming in Ref.\ \cite{Gutoski12} that $H$ is in the unit ball of the dual of the diamond norm if and only if it can be written $H=Q_+-Q_-$ for some interactive measurement $\set{Q_\pm}$, from which the claim follows.
(An alternate proof of this fact that does not use semidefinite programming follows almost immediately from \cite[Theorem~2]{JK13}.)

What is the dual of the diamond norm of a Hermitian operator $H$?
The familiar diamond norm is traditionally defined for linear maps of the form $\Psi:\lin{\cX}\to\lin{\cY}$ so let us begin by considering the dual of the diamond norm for linear maps of this form:
\[ \dnorm{\Psi}^* \defeq \max_{\dnorm{\Phi}\leq 1} \Inner{\Phi}{\Psi}. \]
Here the inner product $\inner{\Phi}{\Psi}$ between two linear maps $\Phi,\Psi:\lin{\cX}\to\lin{\cY}$ is defined in the natural way by
\[ \inner{\Phi}{\Psi} \defeq \sum_{i,j=1}^{\dim(\cX)} \Ptr{}{\Phi(\ket{i}\bra{j})^*\Psi(\ket{i}\bra{j})} = \ptr{}{J(\Phi)^*J(\Psi)}. \]
The diamond norm and its dual can also be defined for operators $H\in\lin{\cY\cX}$ via the Choi--Jamio\l kowski isomorphism.
To this end let $\Psi_H:\lin{\cX}\to\lin{\cY}$ be the unique linear map with $J(\Psi_H)=H$.
The diamond norm and its dual of $H$ are thus given by
\begin{align*}
  \dnorm{H} &\defeq \dnorm{\Psi_H}, \\
  \dnorm{H}^* &\defeq \dnorm{\Psi_H}^*.
\end{align*}

\section{Reduction to finding large spaces} \label{sec:reduction}

Now that we have formally defined the concept of an interactive observable we may return to the question of the number of such observables required to uniquely identify a unitary or low-rank channel.
Previous works on state tomography began by reducing the problem to one of finding large subspaces of traceless Hermitian matrices of high rank \cite{HTW11} or with a prescribed number of positive eigenvalues \cite{CDJJKSZ12}.
Following this lead, we begin by reducing the problem of process tomography to one of finding large subspaces of Hermitian matrices with zero \emph{partial} trace (as opposed to zero trace) and having large rank or a prescribed number of positive eigenvalues.

Let $H_1,\dots,H_m\in\lin{\cY\cX}$ be arbitrary $d^2\times d^2$ Hermitian matrices, which we view as interactive observables for channels acting on a $d$-level quantum system as discussed in Section \ref{sec:observable}, and let $\vec{H}=(H_1,\dots,H_m)$ denote an arbitrary ordering of these observables.
For any Hermitian matrix $X\in\lin{\cY\cX}$ we write
\[ \vec{H}(X)=\Pa{\ptr{}{H_1X},\dots,\ptr{}{H_mX}}\in\bbR^m \]
so that for each channel $\Phi:\lin{\cX}\to\lin{\cY}$ the symbol $\vec{H}(J(\Phi))$ denotes an ordered vector of expectations obtained by measuring each of $H_1,\dots,H_m$ with respect to $\Phi$.
Under this notation, a set $\set{H_1,\dots,H_m}$ of interactive observables uniquely identifies channels in $A$ among channels in $B$ if and only if $\vec{H}(J(\Phi))\neq \vec{H}(J(\Psi))$ for each choice of distinct $\Phi\in A$ and $\Psi\in B$.

Let $\bV\subset\bQ\subset\lin{\cY\cX}$ be Hermitian subspaces of $\lin{\cY\cX}$ such that
\begin{enumerate}

\item $J(\Phi)-J(\Psi)\in\bQ$ for each choice of $\Phi\in A$ and $\Psi\in B$.

\item $J(\Phi)-J(\Psi)\not\in\vaa$ for any choice of distinct channels $\Phi\in A$ and $\Psi\in B$.

\end{enumerate}
A space $\bV$ with these properties called a \emph{discriminating subspace} for $A,B$.
Letting $\vaa^\perp$ denote the
orthogonal compliment of $\vaa$ within $\zpt$,
we claim that any spanning set $\set{H_1,\dots,H_m}$ of $\vaa^\perp$ uniquely identifies channels in $A$ among channels in $B$.

To see this, let $\vec{H}=(H_1,\dots,H_m)$ be any list of interactive observables such that the space $\vaa=\Pa{\spn\set{H_1,\dots,H_m}}^\perp$ is a discriminating subspace for $A,B$.
Then for any distinct channels $\Phi\in A,\Psi\in B$ we have
\begin{align*}
\vec{H}(J(\Phi))=\vec{H}(J(\Psi)) &\iff \vec{H}\Pa{J(\Phi)-J(\Psi)}=\vec{0} \\
&\iff J(\Phi)-J(\Psi)\in\vaa \\
&\implies \textrm{ either $\Phi\not\in A$ or $\Psi\not\in B$}
\end{align*}
from which we conclude that $\set{H_1,\dots,H_m}$ uniquely identifies each channel $A$ among channels in $B$.

Given spaces $\bV,\bQ$ as above, the number $m$ of interactive observables required to identify channels in $A$ among channels in $B$ is given by
\[ m = \dim(\bV^\perp) = \dim(\bQ) - \dim(\bV). \]
This number is minimized when the dimension of $\bQ$ is minimal and the dimension of $\vaa$ is maximal.

It appears that there is little flexibility in the choice of $\bQ$: in this paper we take $\bQ$ to be one of the two spaces
\begin{align}
\bQ_\mathrm{all} &\defeq \spn\Set{ J(\Phi)-J(\Psi) \mid \Phi,\Psi:\lin{\cX}\to\lin{\cY} \textrm{ are channels } } \label{eq:dim-all} \\
&= \Set{ Q \in\lin{\cY\cX} \mid \ptr{\cY}{Q} = 0_\cX } \nonumber \\
\bQ_\mathrm{unital} &\defeq \spn\Set{ J(\Phi)-J(\Psi) \mid \Phi,\Psi:\lin{\cX}\to\lin{\cY} \textrm{ are unital channels } } \label{eq:dim-unital} \\
&= \Set{ Q \in\lin{\cY\cX} \mid \ptr{\cY}{Q} = \ptr{\cX}{Q} = 0 } \nonumber
\end{align}
having
\begin{align*}
  \dim (\bQ_\mathrm{all}) &= d^4-d^2 \\
  \dim (\bQ_\mathrm{unital}) &= d^4-2d^2+1
\end{align*}
The difficulty lies in exhibiting large discriminating subspaces $\bV\subset\bQ$ for $A,B$;
the remainder of this paper is devoted to this task for various choices of sets $A,B$.

\section{Detailed statement of results} \label{sec:results}

We address the following three questions on process tomography:
\begin{quote}
  \emph{How many interactive observables are required in order to uniquely identify...
  \begin{enumerate}
  \item \label{it:question:low-rank} ...channels with at most $q$ Kraus operators among other channels with at most $q$ Kraus operators?
  \item \label{it:question:all} ...channels with at most $q$ Kraus operators among all other channels?
  \item \label{it:question:unital} ...unital channels with at most $q$ Kraus operators among all other unital channels?
  \end{enumerate}}
\end{quote}
In Section \ref{sec:reduction} we reduced questions of this type to the problem of finding large discriminating subspaces.
For each of these three questions we identify a potential class of discriminating subspaces and we assert the explicit constructibility of a subspace in that class.
(The constructions themselves are given in Section \ref{sec:high_rank}.)
Answers to the above questions are given in Section \ref{sec:results:tomography} after we have enumerated our subspace constructions in Sections \ref{sec:results:low-rank}--\ref{sec:results:unital}.

\subsection{Question \ref{it:question:low-rank}: Identification among channels with few Kraus operators}
\label{sec:results:low-rank}

The subspace $\bQ_\mathrm{all}$ of Eq.\ \eqref{eq:dim-all} contains the difference $J(\Phi)-J(\Psi)$ for every choice of channels $\Phi,\Psi:\lin{\cX}\to\lin{\cY}$.
If $\Phi,\Psi$ each have at most $q$ Kraus operators then their Choi matrices $J(\Phi),J(\Psi)$ each have rank at most $q$ and so it must be that $\rank\pa{J(\Phi)-J(\Psi)}\leq 2q$.
Thus, any subspace of $\bV\subset\bQ_\mathrm{all}$ in which every element has rank at least $2q+1$ is a discriminating subspace for question \ref{it:question:low-rank}.

\begin{prop}[Subspaces of high-rank matrices with vanishing partial trace]
\label{thm:low-rank}
  \thmlowrank
\end{prop}

In the special case of question \ref{it:question:low-rank} in which $q=1$ we are asked to identify unitary channels among other unitary channels.
Because every unitary channel is also a unital channel it holds that the difference $J(\Phi)-J(\Psi)$ is contained in the smaller subspace $\bQ_\mathrm{unital}$ of Eq.\ \eqref{eq:dim-unital}.
In this case, a slight improvement can be obtained if we construct a discriminating subspace within $\bQ_\mathrm{unital}$.

\begin{prop}[Subspaces of rank-three matrices with two vanishing partial traces]
\label{thm:unitary}
  \thmunitary
\end{prop}

Proposition \ref{thm:unitary} serves to reduce the number of interactive measurements from $9$ to $6$ in the $d = 2$ case of identifying unitary channels among other unitaries.  Otherwise, it reduces that number only by $1$.

\subsection{Question \ref{it:question:all}: Identification among all channels}
\label{sec:results:all}

As observed in Section \ref{sec:results:low-rank}, the subspace $\bQ_\mathrm{all}$ of Eq.\ \eqref{eq:dim-all} contains the difference $J(\Phi)-J(\Psi)$ for every choice of channels $\Phi,\Psi:\lin{\cX}\to\lin{\cY}$.
If $\Phi$ has at most $q$ Kraus operators then its Choi matrix $J(\Phi)$ has rank at most $q$ and so it must be that $J(\Phi)-J(\Psi)$ has at most $q$ positive eigenvalues for any choice of channel $\Psi$.
Thus, any subspace of $\bV\subset\bQ_\mathrm{all}$ in which every element has at least $q+1$ positive eigenvalues is a discriminating subspace for question \ref{it:question:all}.

\begin{prop}[Subspaces of matrices with many positive eigenvalues and vanishing partial trace]
\label{thm:all}
  \thmall
\end{prop}

\subsection{Question \ref{it:question:unital}: Identification among unital channels}
\label{sec:results:unital}

The subspace $\bQ_\mathrm{unital}$ of Eq.\ \eqref{eq:dim-unital} contains the difference $J(\Phi)-J(\Psi)$ for every choice of unital channels $\Phi,\Psi:\lin{\cX}\to\lin{\cY}$.
If $\Phi$ has at most $q$ Kraus operators then its Choi matrix $J(\Phi)$ has rank at most $q$ and so it must be that $J(\Phi)-J(\Psi)$ has at most $q$ positive eigenvalues for any choice of channel $\Psi$.
Thus, any subspace of $\bV\subset\bQ_\mathrm{unital}$ in which every element has at least $q+1$ positive eigenvalues is a discriminating subspace for question \ref{it:question:unital}.

\begin{prop}[Subspaces of matrices with many positive eigenvalues and two vanishing partial traces]
\label{thm:unital}
  \thmunital
\end{prop}

Indeed, the bound on the dimension of $\bV_{q+,\mathrm{unital}}$ in Theorem \ref{thm:unital} can be improved slightly when $q\geq d$, but it is too cumbersome to write the precise quantity here: see Section \ref{sec:high_rank:unital} for details.

\subsection{The number of interactive observables required for process tomography}
\label{sec:results:tomography}

Our answers to questions \ref{it:question:low-rank}--\ref{it:question:unital} from the beginning of this section are as follows.

\begin{thm}[Process tomography with $O(d^2)$ interactive observables]
\label{thm:unitary_uda_udp}

  The following hold for channels acting on a $d$-level quantum system:
  \begin{enumerate}

  \item \label{it:unitary_uda_udp:unitary}
  $d^4 - d^2 - \dim(\bV_{2q})$ interactive observables suffice to identify all channels with at most $q$ Kraus operators among all other channels with at most $q$ Kraus operators.

  (Here $\dim(\bV_{2q})$ is given in Proposition \ref{thm:low-rank}.)

  In the special case of identifying unitaries among unitaries ($q=1$), Proposition \ref{thm:unitary} further reduces this number to $4d^2-2d-4$ for $d\geq 3$ and from 9 to 6 when $d=2$.

  \item
  $d^4 - d^2 - \dim(\bV_{q+})$ interactive observables suffice to identify all channels with at most $q$ Kraus operators among all other channels.

  (Here $\dim(\bV_{q+})$ is given in Proposition \ref{thm:all}.)

  \item
  $d^4 - 2d^2 + 1 - \dim(\bV_{q+,\mathrm{unital}})$ interactive observables suffice to identify all unital channels with at most $q$ Kraus operators among all other unital channels.

  (Here $\dim(\bV_{q+,\mathrm{unital}})$ is given in Proposition \ref{thm:unital}.)

  \end{enumerate}
\end{thm}

The claims of Section \ref{sec:intro:results} are recovered from Theorem \ref{thm:unitary_uda_udp} in the special case $q=1$.

\section{Notes on total nonsingularity}
\label{sec:total_nonsing}

The remainder of the paper is devoted to the construction of the large subspaces of Hermitian matrices satisfying the various partial trace, rank, and eigenvalue restrictions described in Section~\ref{sec:results}. Our key building block in the construction of such subspaces (and indeed, one of the building blocks used in each of Refs.\ \cite{CW08,HTW11,CDJJKSZ12}) is the notion of \emph{totally nonsingular matrices}---matrices in which every square submatrix is nonsingular.

\subsection{Total nonsingularity and high-rank subspaces}
\label{sec:total_nonsing_highrank}

One of the most well-known examples of a totally non-singular matrix is a \emph{Vandermonde matrix} \cite{Fal01}
\begin{align*}
	\begin{bmatrix}
		1 & \alpha_1 & \alpha_1^2 & \cdots & \alpha_1^{d-1} \\
		1 & \alpha_2 & \alpha_2^2 & \cdots & \alpha_2^{d-1} \\
		\vdots & \vdots & \vdots & \ddots & \vdots \\
		1 & \alpha_d & \alpha_d^2 & \cdots & \alpha_d^{d-1}
	\end{bmatrix}
\end{align*}
in which $0 < \alpha_1 < \alpha_2 < \cdots < \alpha_d$. Although only a few other explicit families of totally nonsingular matrices are known, total nonsingularity is a general phenomenon in the sense that the matrices that are not totally nonsingular form a set of measure zero.

The ways in which totally nonsingular matrices are used to construct high-rank subspaces in Refs. \cite{CW08,HTW11,CDJJKSZ12} are all similar and involve placing the columns of a totally nonsingular matrix along diagonals of matrices. For example, one way to construct a real $9$-dimensional subspace $\bV \subseteq \bbM_4$ of Hermitian matrices such that every nonzero $V \in \bV$ has $\rank(V) \geq 2$ is to construct matrices with columns of Vandermonde matrices down various diagonals as follows (we use $\cdot$ to denote $0$ entries):
\begin{align*}
	D_1 & = \begin{bmatrix}
		1 & \cdot & \cdot & \cdot \\
		\cdot & 1 & \cdot & \cdot \\
		\cdot & \cdot & 1 & \cdot \\
		\cdot & \cdot & \cdot & 1
	\end{bmatrix} & D_2 & = \begin{bmatrix}
		1 & \cdot & \cdot & \cdot \\
		\cdot & 2 & \cdot & \cdot \\
		\cdot & \cdot & 3 & \cdot \\
		\cdot & \cdot & \cdot & 4
	\end{bmatrix} & D_3 & = \begin{bmatrix}
		1 & \cdot & \cdot & \cdot \\
		\cdot & 4 & \cdot & \cdot \\
		\cdot & \cdot & 9 & \cdot \\
		\cdot & \cdot & \cdot & 16
	\end{bmatrix} \\
	D_4 & = \begin{bmatrix}
		\cdot & 1 & \cdot & \cdot \\
		\cdot & \cdot & 1 & \cdot \\
		\cdot & \cdot & \cdot & 1 \\
		\cdot & \cdot & \cdot & \cdot
	\end{bmatrix} & D_5 & = \begin{bmatrix}
		\cdot & 1 & \cdot & \cdot \\
		\cdot & \cdot & 2 & \cdot \\
		\cdot & \cdot & \cdot & 3 \\
		\cdot & \cdot & \cdot & \cdot
	\end{bmatrix} & D_6 & = \begin{bmatrix}
		\cdot & \cdot & 1 & \cdot \\
		\cdot & \cdot & \cdot & 1 \\
		\cdot & \cdot & \cdot & \cdot \\
		\cdot & \cdot & \cdot & \cdot
	\end{bmatrix}.
\end{align*}
Then the $9$ Hermitian matrices $\{H_j\}$ defined by $H_j = D_j$ $(1 \leq j \leq 3)$, $H_j = D_j + D_j^*$ $(4 \leq j \leq 6)$, and $H_j = iD_{j-3} - iD_{j-3}^*$ $(7 \leq j \leq 9)$ span a real subspace in which every nonzero matrix has rank at least $2$. To see this, notice that in every real linear combination of the $H_j$'s, there is an upper-right-most diagonal that is nonzero, and that diagonal has at least $2$ nonzero entries. The corresponding $2 \times 2$ submatrix is thus nonsingular, which implies that the matrix has rank at least $2$.

In order to generalize this argument to different dimensions and ranks, we need a way of generating columns (to be placed down various diagonals of the matrices $\{H_j\}$) with the property that their linear combinations do not have ``too many'' zero entries. The following lemma, which is well-known, shows that totally nonsingular matrices serve this purpose.
\begin{lemma}\label{lem:tr0_sum_cols}
	Let $c \leq r$ be positive integers. If $V \in \bbM_d$ is such that all of its $r \times c$ submatrices have full rank $c$, then any linear combination of $c$ columns of $V$ contains at most $r - 1$ zero entries.
\end{lemma}
\begin{proof}
	Suppose for a contradiction that some linear combination of $c$ columns of $V$ contained $r$ or more zero entries. Then the $r \times c$ submatrix of $V$ whose columns correspond to the $c$ columns in the linear combination and whose rows correspond to $r$ of the zero entries must have rank at most $c - 1$, which is a contradiction.
\end{proof}

\subsection{Variations of total nonsingularity combined with linear constraints}
\label{sec:total_nonsing_constraints}

While totally nonsingular matrices are useful for the construction of subspaces of high-rank matrices, the subspaces we wish to construct must also satisfy certain partial trace conditions. Thus we don't need totally nonsingular matrices themselves, but rather matrices such that certain subsets of their rows sum to $0$, yet their submatrices all have high rank. The following two lemmas serve this purpose---the first lemma when there is one partial trace constraint, the second when there are two.

The first lemma is intuitive enough (albeit slightly technical) that we only provide an ``intuitive'' and not terribly rigorous proof. The second lemma of this section is much less straightforward, so it is proved rigorously using algebraic geometry techniques. The same techniques could be directly adapted to rigorously prove the first lemma.
\begin{lemma}\label{lem:vander_tr0}

    For all $d,k,m \geq 1$, if we define $f(r) := r - {\rm min}(\lfloor r/d \rfloor,k)$ then there exists $V \in \bbM_{dk+m-1}$ with the following two properties:
	\begin{enumerate}

    \item \label{it:vander_tr0:zero-sum}
    For all $0 \leq j < k$, the sum of rows $jd + 1, jd + 2, \ldots, (j+1)d$ of $V$ equals $0$.

	\item \label{it:vander_tr0:full-rank}
    For all $1 \leq r < dk + m$, every $r \times f(r)$ submatrix of $V$ has full rank $f(r)$.

	\end{enumerate}

\end{lemma}

\begin{proof}

Observe that property \ref{it:vander_tr0:full-rank} is satisfied by a generic matrix with property \ref{it:vander_tr0:zero-sum}.
Indeed, suppose first that $m = 1$ and the entries of $V$ are selected uniformly at random from the interval $[0,1]$ subject to the constraint that every $d$th row is the negative of the sum of the previous $d-1$ rows. If $r \geq c$ then, generically, every $r \times c$ submatrix of $V$ will have rank ${\rm min}\{r - \ell,c\}$, where $\ell$ is the number of distinct values of $j$ such that this submatrix contains each of the $d$ rows $jd + 1, jd + 2, \ldots, (j+1)d$ of $V$. Since $\ell \leq \lfloor r/d \rfloor$, the result in the $m = 1$ case follows by letting $c = r - \lfloor r/d \rfloor$. For the case of general $m$, simply append $m-1$ randomly-generated rows to the bottom of $V$.
\end{proof}

To help illustrate Lemma~\ref{lem:vander_tr0}, we note that in the $d = k = 2, m = 1$ case it says that there exist $4 \times 4$ matrices such that the sum of the first two rows equals $0$, the sum of the last two rows equals $0$, every $1 \times 1$ submatrix has rank $1$, and every $3 \times 2$ submatrix has rank $2$ (it also says that every $2 \times 1$ submatrix has rank $1$ and that every $4 \times 2$ submatrix has rank $2$, but these properties follow automatically from the $1 \times 1$ and $3 \times 2$ rank properties). Such a matrix is easily constructed simply by randomly generating the first and third rows, setting the second and fourth rows equal to their negatives, and then verifying that the rank conditions are satisfied. For example, it is easily-verified that the following matrix satisfies all of the requirements:
\begin{align*}
	\begin{bmatrix}
	 1 &  1 &  1 &  1 \\
    -1 & -1 & -1 & -1 \\
     1 &  2 &  3 &  4 \\
    -1 & -2 & -3 & -4
	\end{bmatrix}.
\end{align*}

While Lemma~\ref{lem:vander_tr0} is useful for constructing subspaces of high-rank matrices with vanishing partial trace, we also need to construct subspaces of high-rank matrices with \emph{two} vanishing partial traces.
The following lemma serves this purpose.

\begin{lemma}\label{lem:vander_both_tr0}

	For all $d \geq 1$, if we define $f(r) := {\rm min}(r - \lfloor (r-1)/(d-1) \rfloor, (d-1)^2)$ then there exists $V \in \bbM_{d^2}$ with the following three properties:
	\begin{enumerate}

	\item \label{it:vander-both-tr0:zero-sum}
    For all $0 \leq j < d$, the sum of rows $jd + 1, jd + 2, \ldots, (j+1)d$ of $V$ equals $0$.

	\item \label{it:vander-both-tr0:zero-sum2}
    For all $0 \leq j < d$, the sum of rows $j + 1, j + 1 + d, j + 1 + 2d, \ldots, j + 1 + (d-1)d$ of $V$ equals $0$.

	\item \label{it:vander-both-tr0:full-rank}
    For all $1 \leq r \leq d^2$, every $r \times f(r)$ submatrix of $V$ has full rank $f(r)$.

	\end{enumerate}

\end{lemma}

\begin{proof}

As in the proof of Lemma~\ref{lem:vander_tr0}, observe that a generic matrix that satisfies properties \ref{it:vander-both-tr0:zero-sum} and \ref{it:vander-both-tr0:zero-sum2} also satisfies property \ref{it:vander-both-tr0:full-rank}. This lemma is perhaps less intuitive than Lemma~\ref{lem:vander_tr0} though, so we prove it a bit more formally.
	
We begin by noting that the set of matrices satisfying conditions \ref{it:vander-both-tr0:zero-sum} and \ref{it:vander-both-tr0:zero-sum2} of the lemma forms an irreducible real algebraic variety $\mathcal{Z} \subseteq \mathbb{R}^{2d^2}$ (the fact that it is irreducible follows from it being a subspace). Let $1 \leq i_1 < i_2 < \cdots < i_r \leq d^2$ and $1 \leq j_1 < j_2 < \cdots < j_{f(r)} \leq d^2$ be integers, and define $\mathcal{Z}_{i_1 \ldots i_r, j_1 \ldots j_{f(r)}} \subseteq \mathbb{R}^{2d^2}$ to be the real variety consisting of matrices whose submatrix formed by rows $i_1, i_2 \ldots, i_r$ and columns $j_1, j_2 \ldots, j_{f(r)}$ has rank strictly less than $f(r)$.
	
	Our goal is to show that
	\begin{align}\label{eq:real_vari}
		\mathcal{Z} \setminus \Big( \bigcup\limits_{ \atop{1\leq i_1<\cdots <i_r\leq d^2\atop 1\leq j_1<\cdots<j_{f(r)}\leq d^2}}  \mathcal{Z}_{i_1 \ldots i_r,j_1 \ldots j_{f(r)}}\Big) \neq \emptyset.
	\end{align}
	Suppose for now, for each fixed $1 \leq i_1 < i_2 < \cdots < i_r \leq d^2$ and $1 \leq j_1 < j_2 < \cdots < j_{f(r)} \leq d^2$, that $\mathcal{Z} \cap \mathcal{Z}_{i_1 \ldots i_r,j_1 \ldots j_{f(r)}}$ is a proper subset of $\mathcal{Z}$. It then follows that it in fact has zero measure in $\mathcal{Z}$ (since $\mathcal{Z}$ is irreducible). Since a finite union of measure zero sets again has measure zero, this implies that
	\begin{align*}
		\mathcal{Z} \cap \Big( \bigcup\limits_{ \atop{1\leq i_1<\cdots <i_r\leq d^2\atop 1\leq j_1<\cdots<j_{f(r)}\leq d^2}}  \mathcal{Z}_{i_1 \ldots i_r,j_1 \ldots j_{f(r)}}\Big)
	\end{align*}
	has zero measure in $\mathcal{Z}$, which implies Equation~\eqref{eq:real_vari}.
	
	It thus suffices to show that for each fixed $1 \leq i_1 < i_2 < \cdots < i_r \leq d^2$ and $1 \leq j_1 < j_2 < \cdots < j_{f(r)} \leq d^2$ there exists a matrix $\tilde{V}$ satisfying conditions \ref{it:vander-both-tr0:zero-sum} and \ref{it:vander-both-tr0:zero-sum2} of the lemma such that the submatrix of $\tilde{V}$ corresponding to rows $i_1, i_2 \ldots, i_r$ and columns $j_1, j_2 \ldots, j_{f(r)}$ has full rank $f(r)$.
(Henceforth we denote this submatrix by $\tilde{V}_{i_1\ldots i_r,j_1,\ldots,j_{f(r)}}$.)
The existence of $\tilde{V}$ is straightforward, as the conditions \ref{it:vander-both-tr0:zero-sum} and \ref{it:vander-both-tr0:zero-sum2} impose conditions only on ${\rm max}\{ \lfloor(r-1)/(d-1)\rfloor, r - (d-1)^2\} = r-f(r)$ of the rows of $\tilde{V}_{i_1\ldots i_r,j_1,\ldots,j_{f(r)}}$. Thus we are free to set some $f(r) \times f(r)$ submatrix of $\tilde{V}_{i_1\ldots i_r,j_1,\ldots,j_{f(r)}}$ to be the identity matrix (which ensures that $\tilde{V}$ has full rank $f(r)$) and then fill in the remaining entries of $\tilde{V}$ arbitrarily, subject to constraints \ref{it:vander-both-tr0:zero-sum} and \ref{it:vander-both-tr0:zero-sum2}.
\end{proof}

To help illustrate the above result, we present the following example of a matrix $V \in \bbM_{9}$ that satisfies all three conditions of Lemma~\ref{lem:vander_both_tr0} in the $d = 3$ case:
\begin{align*}
V = \begin{bmatrix}
1 & 1 & 1 & 1 & 1 & 1 & 1 & 1 & 1 \\
1 & 3 & 4 & 5 & 6 & 7 & 8 & 9 & 10 \\
-2 & -4 & -5 & -6 & -7 & -8 & -9 & -10 & -11 \\
1 & 9 & 16 & 25 & 36 & 49 & 64 & 81 & 100 \\
1 & 27 & 64 & 125 & 216 & 343 & 512 & 729 & 1000 \\
-2 & -36 & -80 & -150 & -252 & -392 & -576 & -810 & -1100 \\
-2 & -10 & -17 & -26 & -37 & -50 & -65 & -82 & -101 \\
-2 & -30 & -68 & -130 & -222 & -350 & -520 & -738 & -1010 \\
4 & 40 & 85 & 156 & 259 & 400 & 585 & 820 & 1111
\end{bmatrix}.
\end{align*}
It is straightforward to verify that: (a) the sum of rows $1,2,3$ is zero, as is the sum of rows $4,5,6$ and the sum of rows $7,8,9$; (b) the sum of rows $1,4,7$ is zero, as is the sum of rows $2,5,8$ and the sum of rows $3,6,9$; and (c) every $1 \times 1$, $2 \times 2$, $4 \times 3$, and $6 \times 4$ submatrix of $V$ is nonsingular.

By combining Lemma~\ref{lem:vander_tr0} or~\ref{lem:vander_both_tr0} with Lemma~\ref{lem:tr0_sum_cols}, we are able to construct matrices whose rows sum to zero in ways that play nicely with the partial trace and are also such that linear combinations of their columns contain large numbers of nonzero entries. We make use of such matrices in the next section to construct high-rank subspaces with zero partial trace.

\section{Construction of large subspaces of Hermitian matrices}
\label{sec:high_rank}

In this section we invoke the preliminary results of Section \ref{sec:total_nonsing} on totally nonsingular matrices in order to prove Propositions \ref{thm:low-rank}--\ref{thm:unital} of Section \ref{sec:results} on the existence of large subspaces of Hermitian matrices that satisfy various partial trace, rank, and eigenvalue restrictions.
We repeat the statements of those propositions here for readability.

Before proving these results, it will be useful to recall explicitly that elements of $\lin{\cY\cX}$ can be thought of as $d\times d$ block matrices, and the maps $\ptr{\cX}{\cdot}$ and $\ptr{\cY}{\cdot}$ have very simple formulas in terms of the elements of such block matrices:
\begin{align*}
	\Ptr{\cX}{\begin{bmatrix}
		H_{11} & H_{12} & \cdots & H_{1d} \\
		H_{21} & H_{22} & \cdots & H_{2d} \\
		\vdots & \vdots & \ddots & \vdots \\
		H_{d1} & H_{d2} & \cdots & H_{dd}
	\end{bmatrix}} & = \begin{bmatrix}
		\trace(H_{11}) & \trace(H_{12}) & \cdots & \trace(H_{1d}) \\
		\trace(H_{21}) & \trace(H_{22}) & \cdots & \trace(H_{2d}) \\
		\vdots & \vdots & \ddots & \vdots \\
		\trace(H_{d1}) & \trace(H_{d2}) & \cdots & \trace(H_{dd})
	\end{bmatrix} \quad \text{and} \\
	\Ptr{\cY}{\begin{bmatrix}
		H_{11} & H_{12} & \cdots & H_{1d} \\
		H_{21} & H_{22} & \cdots & H_{2d} \\
		\vdots & \vdots & \ddots & \vdots \\
		H_{d1} & H_{d2} & \cdots & H_{dd}
	\end{bmatrix}} & = \sum_{i=1}^d H_{ii}.
\end{align*}

Since our goal in the proofs of the following propositions is to construct a basis of Hermitian matrices that spans a space of high rank with zero partial trace, it suffices to find such a basis with either $\trace(H_{ij}) = 0$ for all $1 \leq i,j \leq d$ and all $H$ in the basis, or with $\sum_{i=1}^d H_{ii} = 0_\cX$ for all $H$ in the basis, or both, depending on which subset of the partial traces we want to vanish.

\subsection{Subspaces of high-rank matrices with vanishing partial trace}

We now present our results that allow us to construct large subspaces of Hermitian matrices with high rank and zero partial trace. Throughout the proofs of both of these lemmas, we will discuss various diagonals of block matrices, only some of which have any effect on the partial trace. For example, in the $3^2 \times 3^2$ block matrix
\begin{align*}
	\left[\begin{array}{ccc|ccc|ccc}
	9 & 8 & 7 & 6 & 5 & 4 & 3 & 2 & 1 \\
	8 & 9 & 8 & 7 & 6 & 5 & 4 & 3 & 2 \\
	7 & 8 & 9 & 8 & 7 & 6 & 5 & 4 & 3 \\ \hline
	6 & 7 & 8 & 9 & 8 & 7 & 6 & 5 & 4 \\
	5 & 6 & 7 & 8 & 9 & 8 & 7 & 6 & 5 \\
	4 & 5 & 6 & 7 & 8 & 9 & 8 & 7 & 6 \\ \hline
	3 & 4 & 5 & 6 & 7 & 8 & 9 & 8 & 7 \\
	2 & 3 & 4 & 5 & 6 & 7 & 8 & 9 & 8 \\
	1 & 2 & 3 & 4 & 5 & 6 & 7 & 8 & 9		
	\end{array}\right],
\end{align*}
only the diagonals consisting of $7$'s, $8$', and $9$'s affect $\trace_\cY$ (in general, only the main diagonal and the $2d-2$ other nearest diagonals affect $\trace_\cY$), and only the diagonals consisting of $3$'s, $6$'s, and $9$'s affect $\trace_\cX$ (in general, only the main diagonal and every $d$-th diagonal on either side of it affect $\trace_\cX$).

\begin{numberedprop}{\ref{thm:low-rank}}[Subspaces of high-rank matrices with vanishing partial trace]
  \thmlowrank
\end{numberedprop}

\begin{proof}
  We prove the statement by giving an explicit construction, which is built upon the construction given in \cite[Proposition~4]{HTW11} (i.e., the construction that was roughly illustrated in Section~\ref{sec:total_nonsing_highrank}).
  We begin by taking the basis $\{H_i\}$ of the subspace given without the partial trace condition and then impose the partial trace condition and see how much the dimensionality is reduced.
  (However, we actually set $\ptr{\cX}{H}=0_\cY$ rather than $\ptr{\cY}{H}=0_\cX$, as this results in a larger subspace---one can then just swap the spaces $\cX,\cY$ to get $\ptr{\cY}{H}=0_\cX$.)
  More specifically, we consider the $(d^2-2q)^2$ matrices $H_i$ constructed along diagonals in the proof of \cite[Proposition~4]{HTW11} and think of them as $d \times d$ block matrices.
  In order to ensure that $\ptr{\cX}{H_i} = 0$ for all $i$, if $H_i$ is defined along a diagonal that consists of the main diagonals of $k \geq 1$ of its subblocks, we place along that diagonal the columns of a matrix described by Lemma~\ref{lem:vander_tr0} (with $m = 1$) rather than the columns of a totally nonsingular matrix.

  Since we want any linear combination of the $H_i$'s to have at least $2q+1$ nonzero entries, from Lemma~\ref{lem:tr0_sum_cols} and the $m=1$ case of Lemma~\ref{lem:vander_tr0} we can construct $(dk-2q) - \lfloor (dk-2q)/d \rfloor$ such matrices along these ``problem'' diagonals, rather than the $dk-2q$ such matrices in the original proof. We thus lose $\lfloor (dk-2q)/d \rfloor$ matrices along each diagonal that goes through $k$ block diagonals, for a total of
  \begin{align*}
  	\lfloor (d^2-2q)/d \rfloor + 2\sum_{k=1}^{d-1}\lfloor (dk-2q)/d \rfloor = \begin{cases} (d-\lfloor (2q+d-1)/d \rfloor)^2 & \mbox{if } d \text{ is odd} \\
		(d-\lfloor (2q+d-2)/d \rfloor)^2 & \mbox{if } d \text{ is even} \end{cases}
  \end{align*}
  matrices. Since the original subspace without the partial trace condition had dimension $(d^2 - 2q)^2$, the subspace $\bV_{2q}$ described after imposing the partial trace condition has the dimensionality specified in the statement of the proposition.
\end{proof}

\begin{numberedprop}{\ref{thm:unitary}}[Subspaces of rank-three matrices with two vanishing partial traces]
  \thmunitary
\end{numberedprop}

\begin{proof}
  We prove the result by giving an explicit construction that arises from modifying the basis given in the proof of Proposition~\ref{thm:low-rank}. That is, we begin by taking the basis $\{H_i\}$ of the subspace given with only one partial trace restriction and then impose the other partial trace condition and see how much the dimensionality is reduced.

  As before, we think of each $H_i$ as a $d \times d$ block matrix. In order to ensure that $\ptr{\cY}{H_i} = 0_\cX$ for all $i$, if $H_i$ is defined along its main diagonal, we place along that diagonal the columns of a matrix described by Lemma~\ref{lem:vander_both_tr0} rather than the columns of a matrix described by Lemma~\ref{lem:vander_tr0}. Similarly, if $H_i$ is defined along one of the other $2d-2$ diagonals closest to the main diagonal, we place along that diagonal the columns of a matrix described by Lemma~\ref{lem:vander_tr0} (with $m > 1$ and the rows permuted accordingly so that $\ptr{\cY}{H_i} = 0_\cX$) rather than the columns of a totally nonsingular matrix.

  Since we want any linear combination of the $H_i$'s to have at least $3$ nonzero entries, from Lemmas~\ref{lem:tr0_sum_cols} and~\ref{lem:vander_both_tr0} we can construct $(d-1)^2$ such matrices along the main diagonal (i.e., we lose $d-2$ matrices), and along the other $2d-2$ diagonals nearest to the main diagonal we can construct $2(d-1)(d^2-d-2)$ matrices (i.e., we lose $d(d-1)$ matrices along these diagonal). Since the original subspace $\bV_{2}$ with just one partial trace constraint had dimension $(d^2 - 2)^2 - (d-1)^2$, the subspace $\bV_{2,\mathrm{unital}}$ described after imposing the second partial trace constraint has dimension $(d^2 - 2)^2 - (d-1)^2 - (d-2) - d(d-1) = d^4 - 6d^2 + 2d + 5$.

  The claim about the case $d = 2$ follows from observing the the following $3$ block matrices form a basis for a subspace that satisfies all of the given rank and partial trace requirements:
  \begin{align}\label{eq:d2unitaryeg}
	\left[\begin{array}{cc|cc}
	1 & 0 & 0 & 0 \\
	0 & -1 & 0 & 0 \\ \hline
	0 & 0 & -1 & 0 \\
	0 & 0 & 0 & 1
	\end{array}\right], \quad \left[\begin{array}{cc|cc}
	0 & 1 & 0 & 0 \\
	1 & 0 & 2 & 0 \\ \hline
	0 & 2 & 0 & -1 \\
	0 & 0 & -1 & 0
	\end{array}\right], \quad \left[\begin{array}{cc|cc}
	0 & i & 0 & 0 \\
	-i & 0 & 2i & 0 \\ \hline
	0 & -2i & 0 & -i \\
	0 & 0 & i & 0
	\end{array}\right].
  \end{align}
\end{proof}

\subsection{Subspaces of matrices with many positive eigenvalues and vanishing partial trace}
\label{sec:high_rank:unital}

We now present our results that allow us to construct large subspaces of Hermitian matrices with many positive eigenvalues. Unlike the proofs in the previous subsection, the matrices considered here are constructed along \emph{anti}-diagonals (i.e., diagonals that run from lower-left to upper-right) rather than standard diagonals. That is, we now follow the approach of \cite[Lemma~2]{CDJJKSZ12}, rather than that of \cite[Proposition~4]{HTW11}, but our constructions are again more complicated as a result of us requiring that the partial trace of every element of the subspace vanishes.

Note that, in the following propositions, we only consider the upper-triangular part of the anti-diagonals (i.e., the part of the anti-diagonals lying \emph{strictly} above the main diagonal), as the Hermiticity condition determines the lower-triangular entries and all entries along the main diagonal equal $0$.

\begin{numberedprop}{\ref{thm:all}}[Subspaces of matrices with many positive eigenvalues and vanishing partial trace]
  \thmall
\end{numberedprop}

\begin{proof}
  We prove the statement by giving an explicit construction of such a subspace, which is built upon the construction given in \cite[Lemma~2]{CDJJKSZ12}. In fact, our construction is exactly the same, but with just one change in order to impose the partial trace condition.

  As usual, we think of each matrix $V \in \bV_{q+}$ as a $d \times d$ block matrix. We construct a basis $\{H_i\}$ of $\bV_{q+}$ as in \cite{CDJJKSZ12} by placing columns of a totally nonsingular matrix along the upper-triangular part of anti-diagonals. In order to ensure that $\ptr{\cY}{H_i} = 0_\cX$ for all $i$, however, we subtract $\ptr{\cY}{H_i}$ from the top-left block of $H_i$. Two examples of such matrices $H_i$ in the $d = 3$ case are as follows:
	\begin{align*}
		\left[\begin{array}{ccc|ccc|ccc}
		0 & 0 & -1 & 0 & 0 & 0 & 0 & 0 & 4 \\
		0 & 0 & 0 & 0 & 0 & 0 & 0 & 3 & 0 \\
		-1 & 0 & 0 & 0 & 0 & 0 & 2 & 0 & 0 \\ \hline
		0 & 0 & 0 & 0 & 0 & 1 & 0 & 0 & 0 \\
		0 & 0 & 0 & 0 & 0 & 0 & 0 & 0 & 0 \\
		0 & 0 & 0 & 1 & 0 & 0 & 0 & 0 & 0 \\ \hline
		0 & 0 & 2 & 0 & 0 & 0 & 0 & 0 & 0 \\
		0 & 3 & 0 & 0 & 0 & 0 & 0 & 0 & 0 \\
		4 & 0 & 0 & 0 & 0 & 0 & 0 & 0 & 0		
		\end{array}\right] & & \left[\begin{array}{ccc|ccc|ccc}
		0 & 0 & 0 & 2 & 0 & 0 & 0 & 0 & 0 \\
		0 & 0 & 0 & 0 & 0 & 0 & 0 & 0 & 0 \\
		0 & 0 & 0 & 0 & 0 & 0 & 0 & 0 & 0 \\ \hline
		2 & 0 & 0 & 0 & 0 & 0 & 0 & 0 & 0 \\
		0 & 0 & 0 & 0 & 0 & 0 & 0 & 0 & 0 \\
		0 & 0 & 0 & 0 & 0 & 0 & 0 & 0 & 0 \\ \hline
		0 & 0 & 0 & 0 & 0 & 0 & 0 & 0 & 0 \\
		0 & 0 & 0 & 0 & 0 & 0 & 0 & 0 & 0 \\
		0 & 0 & 0 & 0 & 0 & 0 & 0 & 0 & 0		
		\end{array}\right]
	\end{align*}

	In the left matrix, we initially placed $[1,2,3,4]$ along one of the upper-triangular anti-diagonals, then we adjusted the other entries to enforce Hermiticity and the partial trace constraint. In the right matrix, we similarly placed $[1,2]$ along one of the upper-triangular anti-diagonals, but in this case, when we enforced the Hermiticity and partial trace constraints, the $(2,3)$- and $(3,2)$-entries of the matrix were forced to equal $0$.
	
  The rightmost example above illustrates the one and only case where forcing $\ptr{\cY}{H_i} = 0_\cX$ has any effect on the proof that these matrices are linearly independent and every matrix in this subspace has at least $q+1$ positive and $q+1$ negative eigenvalues: when the anti-diagonal that defines $H_i$ goes through the top-left $d\times d$ sub-block. In this case, all entries within that block of $H_i$ are forced to be $0$, which reduces the number of distinct matrices $H_i$ that we can construct along that anti-diagonal (for example, the matrix on the right above is \emph{not} included in the basis of our subspace when $q = 1$, since the zeros at the $(2,3)$- and $(3,2)$-entries cause it to only have $1$ positive and $1$ negative eigenvalue). There are $d(d-1)/2$ upper-triangular anti-diagonals that intersect the top-left block. However, ${\rm min}(2q,d(d-1)/2)$ of those anti-diagonals are of length $\leq q$ and do not have an associated matrix $H_i$. For each of the other anti-diagonals that intersect the top-left block, we lose $2{\rm min}(z,\ell-q)$ matrices, where $z$ is the number of entries along that upper-triangular anti-diagonal within the top-left block and $\ell$ is the length of that anti-diagonal. Thus we remove a total of $(d-q-1){\rm max}\{d-q,0\}$ matrices from the set $\{H_i\}$ in order to not violate the eigenvalue requirement. Since the dimension of the subspace constructed in this way in \cite{CDJJKSZ12} \emph{without} the partial trace restriction was $d^4 - (4q+1)d^2 + (4q^2 + 2q)$, the result follows.
\end{proof}

Throughout the proof of the following proposition, we will repeatedly refer to the quantity
\begin{align*}
  	L_{d,k} \defeq \begin{cases} \lfloor \frac{k+1}{2}\rfloor & \mbox{if } k < d \\
																 \lfloor \frac{2d-k-1}{2}\rfloor & \mbox{if } k \geq d \end{cases}.
\end{align*}
That is, $L_{d,k}$ is the length of the $k$th upper-triangular anti-diagonal of a $d \times d$ matrix, just as in \cite{CDJJKSZ12}. Note that we count the upper-triangular anti-diagonals from left to right, so that the $(1,2)$-entry is the first upper-triangular anti-diagonal, the $(1,3)$-entry is the second upper-triangular anti-diagonal, the $(1,4)$- and $(2,3)$-entries together make up the third upper-triangular anti-diagonal, and so on. For a $d \times d$ matrix, there are $2d - 3$ distinct upper-triangular anti-diagonals.

\begin{numberedprop}{\ref{thm:unital}}[Subspaces of matrices with many positive eigenvalues and two vanishing partial traces]
  \thmunital
\end{numberedprop}

Before proving this proposition, we note that the bound on $\dim(\bV_{q+,\mathrm{unital}})$ that we actually prove is slightly better than stated above when $q \geq d$. Indeed, if we define
\begin{align}\label{def:iq}
	\mathcal{I}_q \defeq \big\{ k \in \{1,2,\ldots,2d-3\} : L_{d^2,dk} \geq q+1 \big\}
\end{align}
then we show that there exists such a subspace $\bV_{q+,\mathrm{unital}}$ with $\dim(\bV_{q+,\mathrm{unital}}) = \dim(\bV_{q+}) - 2\sum_{k \in \mathcal{I}_q}L_{d,k}$ when $q \geq d$, and it is straightforward to show that $2\sum_{k \in \mathcal{I}_q}L_{d,k} \leq d^2 - d - 2$ in this case.

\begin{proof}
  Once again, we prove this result by explicitly constructing a basis of such a subspace. Our construction uses the basis $\{H_i\}$ constructed in the proof of Proposition~\ref{thm:all} as a starting point. As before, we think of these matrices as $d \times d$ block matrices. To enforce the requirement that $\ptr{\cX}{H_i} = 0_\cY$ for all $i$, we add to the top-left entry of each block of $H_i$ the negative of the trace of that block. Applying this procedure to the same two matrices considered in the $d=3$ case in the proof of Proposition~\ref{thm:all} results in the following two matrices:
	\begin{align*}
		\left[\begin{array}{ccc|ccc|ccc}
		0 & 0 & -1 & 0 & 0 & 0 & -3 & 0 & 4 \\
		0 & 0 & 0 & 0 & 0 & 0 & 0 & 3 & 0 \\
		-1 & 0 & 0 & 0 & 0 & 0 & 2 & 0 & 0 \\ \hline
		0 & 0 & 0 & 0 & 0 & 1 & 0 & 0 & 0 \\
		0 & 0 & 0 & 0 & 0 & 0 & 0 & 0 & 0 \\
		0 & 0 & 0 & 1 & 0 & 0 & 0 & 0 & 0 \\ \hline
		-3 & 0 & 2 & 0 & 0 & 0 & 0 & 0 & 0 \\
		0 & 3 & 0 & 0 & 0 & 0 & 0 & 0 & 0 \\
		4 & 0 & 0 & 0 & 0 & 0 & 0 & 0 & 0		
		\end{array}\right] & & \left[\begin{array}{ccc|ccc|ccc}
		0 & 0 & 0 & 0 & 0 & 0 & 0 & 0 & 0 \\
		0 & 0 & 0 & 0 & 0 & 0 & 0 & 0 & 0 \\
		0 & 0 & 0 & 0 & 0 & 0 & 0 & 0 & 0 \\ \hline
		0 & 0 & 0 & 0 & 0 & 0 & 0 & 0 & 0 \\
		0 & 0 & 0 & 0 & 0 & 0 & 0 & 0 & 0 \\
		0 & 0 & 0 & 0 & 0 & 0 & 0 & 0 & 0 \\ \hline
		0 & 0 & 0 & 0 & 0 & 0 & 0 & 0 & 0 \\
		0 & 0 & 0 & 0 & 0 & 0 & 0 & 0 & 0 \\
		0 & 0 & 0 & 0 & 0 & 0 & 0 & 0 & 0		
		\end{array}\right].
	\end{align*}
	
	In the left matrix, the only change is that the $(1,7)$- and $(7,1)$-entries of the matrix now equal $-3$. The right matrix, however, now consists entirely of zeroes. As before, there is only one case where this procedure affects the proof, and that is when the anti-diagonal that defines $H_i$ goes through the top-left corner of a block (as in the example on the right above). In this case, that entry of $H_i$ is forced to be $0$, which reduces the number of distinct matrices $H_i$ that we can construct along that anti-diagonal by $2$ (one real matrix and one imaginary matrix).
	
	There are $d(d-1)/2$ top-left corners of blocks in the upper-triangular portion of $H$, however we only remove a matrix from our basis if the anti-diagonal going through that top-left corner has length $\geq q+1$. In other words, we remove $2\sum_{k \in \mathcal{I}_q}L_{d,k}$ matrices, where $\mathcal{I}_q$ is as in~\eqref{def:iq}. However, if $q < d$ then the anti-diagonals hitting the top-left corner of one of these blocks (the $(1,2)$-block) have already been removed by the requirement that $\ptr{\cY}{H_i} = 0_\cX$ from Proposition~\ref{thm:all}, so we add back in $2$ matrices. Furthermore, it is straightforward to verify that $2\sum_{k \in \mathcal{I}_q}L_{d,k} = d^2 - d$ when $q < d$ and $2\sum_{k \in \mathcal{I}_q}L_{d,k} \leq d^2 - d - 2$ when $q \geq d$, so the result follows.
\end{proof}

\section{Directions for future research}\label{sec:future}

\subsection{Generalization to arbitrary affine spaces of states}\label{sec:gen_subspace}

The questions considered in this work can be viewed as questions about state tomography on certain subsets of states via the Choi--Jamio\l kowski isomorphism.
For example, using process tomography to uniquely identify unitary channels among unitary channels is equivalent to using state tomography to uniquely identify maximally entangled bipartite pure states among maximally entangled bipartite pure states.
In other words, we can think of process tomography as state tomography on the intersection of the set of pure states and the affine space of operators with completely mixed partial trace.

This line of thinking raises a natural question: can we answer uniqueness questions for tomography when we restrict attention to affine spaces other than the space of operators with completely mixed partial trace? We do not have a detailed answer for this question, but we can make some observations.

First, some work in this area has already been done:
it was shown in Ref.\ \cite[Proposition~9]{HTW11} that for any manifold $\cl{P}$ with real dimension $d(\cl{P})$ almost all families of $2d(\cl{P})+1$ observables uniquely determine elements of $\cl{P}$ among elements of $\cl{P}$.

For example, the space of unitary channels acting on $\bbM_d$ is a real manifold of dimension $d^2-1$, so almost any family of $2d^2-1$ interactive observables uniquely identifies unitary channels among unitary channels. Remarkably, this quantity is smaller than the $4d^2 - 2d - 4$ interactive observables that were shown to suffice in item \ref{it:unitary_uda_udp:unitary} of Theorem \ref{thm:unitary_uda_udp}! However, our proof of Theorem~\ref{thm:unitary_uda_udp} is constructive; no explicit set of $2d^2-1$ interactive observables is known to uniquely identify unitary channels among other unitary channels.

Second, the above technique does not tell us anything about how many measurements can be used to uniquely determine elements of $\cl{P}$ (such as the set of unitary channels) among a \emph{different} set (such as the set of \emph{all} quantum channels). To get general results in this direction, we simply use the fact that if $\bQ,\bV \subseteq \bbM_{d^2}$ are two subspaces, then $\dim(\bQ\cap\bV) = \dim(\bQ)+\dim(\bV)-\dim(\bQ+\bV)$.

To illustrate this approach, let $\bV$ be the subspace of operators with at least $2$ positive eigenvalues of dimension $d^4 - 5d^2 + 6$ constructed in Ref.\ \cite{CDJJKSZ12} and let $\bQ=\bQ_\mathrm{all}$ be the subspace \eqref{eq:dim-all} of operators with vanishing partial trace of dimension $d^4 - d^2$. Using the fact that $\dim(\bQ_\mathrm{all}+\bV)\leq d^4$, we see that
\begin{align*}
  \dim(\bQ_\mathrm{all}\cap\bV) &\geq (d^4 - 5d^2 + 6) + (d^4 - d^2) - d^4\\
  &= d^4 - 6d^2 + 6
\end{align*}
It then follows from the reduction of Section \ref{sec:reduction} that there is a set of $(d^4 - d^2) - (d^4 - 6d^2 + 6) = 5d^2 - 6$ interactive observables that uniquely identify unitary channels among all channels, which is slightly worse than the $5d^2 - 3d - 4$ interactive observables that were shown in the present paper to suffice for this task.

\subsection{Process tomography of no-signaling channels}\label{sec:future:ns}

A channel $\Phi:\lin{\cX_1\cX_2}\to\lin{\cY_1\cY_2}$ acting on two $d$-level quantum systems is \emph{no-signaling} if the output of the channel on each system is independent of the input to the channel on the other system.
This condition is nicely characterized by the linear constraints
\[ \ptr{\cY_b}{J(\Phi)}=Q_b\ot I_{\cX_b} \textrm{ for some $Q_b$, for each $b\in\set{1,2}$} \]
on the Choi matrix $J(\Phi)$.
Thus, the reduction of Section \ref{sec:reduction} can be used to bound the number of interactive observables needed to identify channels among no-signaling channels by finding discriminating subspaces $\bV$ of $\bQ_\mathrm{no\textrm{-}sig}$ where
\[ \bQ_\mathrm{no\textrm{-}sig} \defeq \spn\Set{ J(\Phi)-J(\Psi) \mid \Phi,\Psi \textrm{ are no-signaling channels}}. \]

However, the task of constructing large subspaces of high-rank matrices within $\bQ_\mathrm{no\textrm{-}sig}$ is considerably more complicated than that of constructing similar spaces within $\bQ_\mathrm{all}$ or $\bQ_\mathrm{unital}$.
Whereas $\bQ_\mathrm{all},\bQ_\mathrm{unital}$ have only one or two vanishing partial trace constraints, $\bQ_\mathrm{no\textrm{-}sig}$ has several additional partial trace constraints.

Our experience is that each additional constraint adds considerably to the complication of constructing high-rank subspaces based on totally nonsingular matrices.
While this method could in principle yield a significant improvement in the number of interactive observables needed to identify no-signaling channels, the question seems to call out for a different approach.

\subsection{Tomography of multi-round quantum strategies}\label{sec:future:strategies}

A \emph{quantum strategy} is a specification of the actions of one party in a multi-party interaction involving the exchange of multiple rounds of quantum messages among the parties.
Channels arise as a special case of strategies in which only one round of messages is exchanged.
The strategy formalism was introduced in Refs.\ \cite{GutoskiW07,ChiribellaD+08a,ChiribellaD+09a} wherein it was shown that a matrix
$Q\in\lin{\cY_1\cdots\cY_r\cX_1\cdots\cX_r}$ represents an $r$-round strategy for input spaces $\cX_1,\dots,\cX_r$ and output spaces $\cY_1,\dots,\cY_r$ if and only if $Q$ is positive semidefinite and
\[ \ptr{\cY_r\cdots\cY_i}{Q} = Q_i \ot I_{\cX_r\cdots\cX_i} \textrm{ for some $Q_i$, for each $i=1,\dots,r$}. \]
Quantum strategies are therefore intimately related to no-signaling channels in that the above constraints also characterize the Choi matrices of channels of the form
\[ \Phi:\lin{\cX_1\cdots\cX_r}\to\lin{\cY_1\cdots\cY_r} \]
that act on $r$ distinct $d$-level quantum systems and are no-signaling from system $j$ to system $i$ for each $1\leq i<j\leq r$.

An important difference between strategies and no-signaling channels (or any other type of channel) is that channel tomography is achieved by preparing an input state and measuring the output.
By contrast, tomography for $r$-round quantum strategies introduces the need to prepare $r$ input systems and process $r$ output systems in sequence, with future inputs possibly depending on previous outputs.
In other words, the ``observables'' required for strategy tomography are themselves full-blown, multi-round strategies.

At first glance this complication might seem intractable.
In order to perform tomography on a multi-round strategy one must deduce not only how the strategy acts on various input states, but also how it acts when various \emph{channels} are applied to its outputs and then subsequently returned to it as inputs in future rounds.

However, strategies that measure other strategies have appeared previously in the literature under the names \emph{$r$-round measuring co-strategy} \cite{GutoskiW07} and \emph{$r$-tester} \cite{ChiribellaD+09a}.
The discussion of Section \ref{sec:observable} generalizes readily in light of this prior work.
Specifically, one can show that:
\begin{enumerate}

\item An arbitrary $d^{2r}\times d^{2r}$ Hermitian matrix $H$ can be viewed as an ``$r$-round interactive observable'' for a strategy that exchanges a $d$-level quantum system $r$ times in sequence.

\item If $H$ is small enough then one can extract from $H$ a collection of objects---an initial state, a sequence of channels with memory, and a measurement for the final system---that describe the behaviour of the observable operationally.

\item Such an extraction is possible if and only if $H$ lies in the unit ball of the \emph{dual of the strategy $r$-norm} presented in Ref.\ \cite{Gutoski12}.

\end{enumerate}
Thus, as with no-signaling channels in Section \ref{sec:future:ns}, one can employ the reduction of Section \ref{sec:reduction} and search for large discriminating subspaces of $\bQ_r$ where
\[ \bQ_r \defeq \spn\Set{Q-Q' \mid Q,Q' \textrm{ are $r$-round strategies}}. \]

In short: the task of formalizing the notion of an $r$-round interactive observable is not the bottleneck that thwarts attempts to bound the number of observables needed for strategy tomography.
Rather, the bottleneck lies in the same place as with no-signaling channels and other types of channel---in the ability (or lack thereof) to construct large subspaces of high-rank matrices that also satisfy a myriad of partial trace constraints.

\subsection{Experimentally-friendly interactive observables}\label{sec:future:experimental}

One of the drawbacks of our approach to process tomography is that the interactive observables obtained from Theorem \ref{thm:unitary_uda_udp} do not necessarily have a nice form---they might not be Clifford operators, they could be non-local, and so on. Hence, they might be difficult to implement in the laboratory. It would be interesting to consider the number of interactive observables required to reconstruct unitary channels under certain additional ``niceness'' restrictions such as these.

In order to illustrate the type of result that would be desirable we consider the $d = 2, q = 1$ case of Theorem~\ref{thm:unitary_uda_udp}.1, which says that six interactive observables suffice to uniquely determine unitary qubit channels among unitary qubit channels. We now show that we can in fact find six interactive observables for this purpose that are both local and Clifford.

To this end our goal is to find six linearly independent Hermitian local Clifford operators $O_1,\ldots,O_6 \in \bbM_2\ot\bbM_2$ with two vanishing partial traces that are each orthogonal to the three operators~\eqref{eq:d2unitaryeg} presented in the proof of Proposition~\ref{thm:unitary}. The following are six such operators:
\begin{align*}
	O_1 & = X \otimes Z & O_2 & = H \otimes Y & O_3 & = (1-i)S^\dagger HS \otimes SX \\
	O_4 & = Y \otimes Z & O_5 & = ZHZ \otimes X & O_6 & = (1-i)SHS^\dagger \otimes XS,
\end{align*}
where $X,Y,Z$ are the Pauli operators and $H,S$ are the Clifford operators with standard basis representations
\begin{align*}
	X & = \begin{bmatrix}
		0 & 1 \\ 1 & 0
	\end{bmatrix},
	& Y & = \begin{bmatrix}
		0 & -i \\ i & 0
	\end{bmatrix},
	& Z & = \begin{bmatrix}
		1 & 0 \\ 0 & -1
	\end{bmatrix},
	& H & = \frac{1}{\sqrt{2}}\begin{bmatrix}
		1 & 1 \\ 1 & -1
	\end{bmatrix},
	& S & = \begin{bmatrix}
		1 & 0 \\ 0 & i
	\end{bmatrix}.
\end{align*}
This construction of $O_1,\ldots,O_6$ is \emph{ad hoc}, and it is not clear whether it can be generalized to the other cases of Theorem~\ref{thm:unitary_uda_udp}.

As described in Section~\ref{sec:observable:decompose}, each of these interactive observables specifies a measurement on a maximally entangled input state, which is not nice.
It is not difficult to see that these interactive observables could be measured using only experimentally-friendly product state inputs.
However, this product state implementation essentially consists of measuring all twelve degrees of freedom and then throwing away half of the resulting information in order to compile the results into the six observables required for tomography.
Obviously, nothing is gained by such an exercise.

\section*{Acknowledgements}

The authors thank Steven T.~Flammia, Marcus da Silva, John Watrous, and Bei Zeng for helpful conversations.
Research at the Perimeter Institute is supported by the Government of Canada through Industry Canada and by the Province of Ontario through the Ministry of Research and Innovation.
GG also acknowledges support from CryptoWorks21.
NJ is supported by the Natural Sciences and Engineering Research Council of Canada.

\bibliographystyle{alpha}
\newcommand{\etalchar}[1]{$^{#1}$}

\end{document}